\documentclass[journal,doublecolumn,10pt]{IEEEtran}

\usepackage{epsfig,latexsym}
\usepackage{float}
\usepackage{indentfirst}
\usepackage{amsmath}
\usepackage{amssymb}
\usepackage{times}
\usepackage{subfigure}
\usepackage{psfrag}
\usepackage{cite}
\usepackage{lastpage}
\usepackage{fancyhdr}

\newtheorem{Corollary}{Corollary}
\newtheorem{theorem}{$\mathbf{Theorem}$}

\newtheorem{corollary}[Corollary]{$\mathbf{Corollary}$}

\begin{document}%
\title{ Cooperative Energy Harvesting Networks with Spatially Random Users}

\author{ Zhiguo Ding, \IEEEmembership{Member, IEEE} and H. Vincent Poor, \IEEEmembership{Fellow, IEEE}
\thanks{  The authors are with Department of Electrical Engineering, Princeton University, Princeton, NJ, 08544, USA. Z. Ding is also  with  School of
Electrical, Electronic, and Computer Engineering Newcastle
University, UK.    }} \vspace{-4em}\maketitle\vspace{-7em}
\begin{abstract}
This paper  considers a cooperative network with multiple source-destination pairs and one energy harvesting relay. The outage probability experienced by users in this network is characterized by taking the spatial randomness of user locations into consideration. In addition, the cooperation among users is modeled as a canonical coalitional game and the grand coalition is shown to be stable in the addressed scenario. Simulation results are provided to demonstrate the accuracy of the developed analytical results.
\vspace{-2em}
\end{abstract}
\section{Introduction}
Simultaneous power and signal transfer has been  recognized as a promising energy harvesting technique, and it is particularly  important to various energy constrained wireless networks without access to natural light or wind sources \cite{varshney08,Grover10,Zhouzhang13}. Wireless power and signal transfer has been studied in multiple antenna systems in \cite{Zhangruipower2013}, and its extension to the case with imperfect channel state information (CSI) at the transmitter has been considered in \cite{Xiangzt12}. This new concept of energy harvesting is ideal for cooperative communication  networks, in which the relay transmissions can be powered by the energy harvested from the incoming signals \cite{Nasirzhou}.

In this paper, we consider a cooperative network with multiple source-destination pairs communicating with each other via an energy harvesting relay. The contribution of this paper is two-fold. {\it Firstly} the spatial randomness of user locations is taken into consideration when the outage probability experienced by users is characterized. In the context of wireless power transfer, this   randomness is particularly important since it determines the distance between the transceivers and hence describes the energy attenuation of the transmitted  signals.  {\it Secondly} the cooperation among users is modeled as a canonical coalitional game, and we show that in the high SNR regime, a grand coalition is always preferred, which means when forming a larger cooperative group the users cannot do worse than by acting alone. Simulation results are provided to demonstrate the accuracy of the developed analytical results.
\section{Cooperative Energy Harvesting Transmissions}
Consider a scenario with $N$ pairs of sources and destinations communicating  with each other via an {\it energy harvesting} relay. Particularly the locations of $2N$ nodes   are distributed uniformly in a disc, denoted by $\mathcal{D}$, with the relay located at its origin and $D$ as its radius.  Let  $h_i$ denote  the channel gain between the $i$-th source and the relay, which we assume to be complex Gaussian (i.e., we assume Rayleigh fading), $d_i$ denote the corresponding distance, and $\alpha$ denote the path loss factor. Similarly  $g_i$   and  $c_i$ are defined for the channels of the destinations.   Prior to transmissions, no CSI, including channel gains and   distances, is available.

The set of all pairs of nodes is denoted by $\mathcal{N}$. Prior to transmissions, user pairs who are willing to cooperate with each other form coalitions.  Let $\mathcal{S}_k\in \mathcal{N}$ denote a coalition consisting of $|\mathcal{S}_k|$ source-destination pairs. $\mathcal{S}\triangleq\{\mathcal{S}_1, \cdots, \mathcal{S}_K\}$ is a partition of $\mathcal{N}$, if  $\mathcal{N}=\cup^{K}_{k=1}\mathcal{S}_k$ and $\mathcal{S}_i\cap \mathcal{S}_j=0$, for any $i\neq j$. Denote the set of all possible partitions by $\mathcal{B}$.   Define $P_i^{ow}$ as  the transmission power of the $i$-th source, and $\tilde{P}^{ow}_i$ as the  relay transmission power allocated to the $i$-th destination.   The cooperative transmission strategy is described as follows:
\begin{itemize}
\item Phase I: Each source sends its message to the relay, and the sources of the same coalition   cooperate  in the sense that their power will be shared under the total power constraint, i.e. $\sum_{i\in \mathcal{S}_k}P_i^{ow}=P^{total}\triangleq |\mathcal{S}_k|P^{ow}$. Note that the sources communicate with the relay via orthogonal channels, and there is no direct source-destination  link.

\item Phase II: For each received transmission, the relay first tries to decode the message, and then carries out energy harvesting if there is any power left after decoding. Then the relay forwards  the correctly decoded messages to the destinations. Let $\hat{P}^{ow}_i$ denote  the power that is harvested from the $i$-th source message.  The user pairs in  the same coalition are cooperating in the sense that their relaying power will be shared under the total harvested energy    constraint, i.e. $\sum_{i\in \mathcal{S}_k}\tilde{P}_i^{ow}=\sum_{i\in \mathcal{S}_k}\hat{P}_i^{ow}$.
\end{itemize}
$\hat{P}_i^{ow}$ is calculated based on the energy harvested from the $i$-th source message, denoted by ${E}_i$, i.e. $\hat{P}_i^{ow}=\frac{ {E}_i}{\frac{T}{2}}$, where $T$ is the duration of one cooperative frame.  For each observation, the signal-to-noise ratio (SNR) threshold for correct  decoding is $\epsilon\triangleq \left(2^{2R}-1\right)$ to ensure $\frac{1}{2}\log (1+SNR)>R$, where $R$ denotes the targeted data rate. Any power left after decoding  will be used for energy harvesting, i.e. $E_i=\frac{T\eta}{2}  \left(\frac{|h_i|^2}{1+d_i^\alpha}P_i^{ow} - \epsilon\right)$, where $\eta$ denotes the energy harvesting efficiency.

  Therefore, the relationship between $\hat{P}_i^{ow}$ and $P_i^{ow}$ depends on the channel gain and distance between the $i$-th source and the relay, as described in the following:
\begin{equation}\label{eq power}
\hat{P}_i^{ow} =\left\{\begin{array}{ll}\eta  \left(\frac{|h_i|^2}{1+d_i^\alpha}P_i^{ow} - \epsilon\right), &\mathrm{if}\quad \frac{|h_i|^2}{1+d_i^\alpha}P_i \geq \epsilon   \\ 0, &\mathrm{if}\quad \frac{|h_i|^2}{1+d_i^\alpha}P_i < \epsilon \end{array}\right..
\end{equation}

Let $\mathrm{P}_i$ denote the outage probability experienced by the $i$-th user pair, and therefore the achieved data rate between this pair is $R(1-\mathrm{P}_i)$. A natural choice for the coalition value of   $\mathcal{S}_k$ is $
v(\mathcal{S}_k) = \sum_{i\in \mathcal{S}_k} R\cdot (1-\mathrm{P}_i )$.
The questions to be studied in this paper are how many coalitions should be formed, and how   a user should choose a coalition to join.   Such questions can be formulated in the following optimization problem:
\begin{eqnarray}\label{original problem}
&\underset{P_i, \mathcal{S}\in \mathcal{B}}{\mathrm{\max} } &\sum_{\mathcal{S}_k\in \mathcal{S}} \left(\sum_{i\in \mathcal{S}_k} R\cdot(1-\mathrm{P}_i ) \right)\\ \nonumber &\mathrm{s.t.}& \sum_{i\in \mathcal{S}_k}P_i^{ow} = |\mathcal{S}_k| P^{ow}, \sum_{i\in  {\mathcal{S}}_k}\tilde{P}_i^{ow}= \sum_{i\in \mathcal{S}_k}\hat{P}_i^{ow}.
\end{eqnarray}

\section{Canonical Coalitional Game}
The optimization problem shown in \eqref{original problem} can be   simplified   by taking the following steps.
\begin{itemize}
\item Since no CSI  is known prior to transmissions, all user pairs are equally important, which means  that equal power allocation is optimal, i.e. $P_i^{ow} = P^{ow}$.

\item Let $\tilde{\mathcal{S}}_k$ denote the subset of $\mathcal{S}_k$ containing all the sources whose information is successfully delivered to the relay.  When no CSI is available  at the relay, all   users are equally important, so equal power allocation is optimal, i.e. $\tilde{P}_i^{ow}=\eta|\tilde{\mathcal{S}}_k|^{-1} \sum_{j\in\tilde{\mathcal{S}}_k} \left(\frac{|h_j|^2}{1+d_j^\alpha}P_j^{ow} - \epsilon\right)$.
\end{itemize}
By applying these steps, an equivalent form of \eqref{original problem} can be expressed as follows:
{\small
\begin{eqnarray}\label{revised problem}
&\underset{P_i, \mathcal{S}\in \mathcal{B}}{\mathrm{\max} } &\sum_{\mathcal{S}_k\in \mathcal{S}} \left(\sum_{i\in \mathcal{S}_k} R\cdot(1-\mathrm{P}_i ) \right)\\ \nonumber &\mathrm{s.t.}& P_i^{ow} =  P^{ow}, \tilde{P}_i^{ow}=\frac{\eta}{|\tilde{\mathcal{S}}_k|} \sum_{j\in\tilde{\mathcal{S}}_k} \left(\frac{|h_j|^2}{1+d_j^\alpha}P_j^{ow} - \epsilon\right).
\end{eqnarray}}

Evaluating  the outage probability is the key to finding  the solution of \eqref{revised problem}, as discussed in the next section. The path loss factor $\alpha$ is set as $2$, in order to obtain closed-form expressions.  \vspace{-1.2em}

\subsection{Characterizing $\mathrm{P}_i$ }
The following theorem provides the high SNR approximations for the outage probability.
\begin{theorem}\label{theorem1}{\it
In the high SNR regime,  the outage probability experienced by  the $i$-th user pair, $i\in\mathcal{S}_k$, can be approximated  as follows:
\begin{align} \label{high snr}
\mathrm{P}_i &\approx  F_{x_i}\left(\frac{\epsilon}{P^{ow}}\right)+ \sum^{|\mathcal{S}_k|}_{n=1}\frac{n}{| {\mathcal{S}}_k|} \frac{ a^{n}  }{(n-1)! }  \left((n-1)!a^{-n} \right. \\ \nonumber &-\frac{2}{D^2\varpi}  \left(\frac{\varpi}{a}\right)^{\frac{n+1}{2}} \mathrm{K_{n+1}}\left(2\sqrt{\varpi a}\right)+\frac{2}{D^2\varpi}\\ \nonumber &\times\left.  \left(\frac{(1+D^2)\varpi}{a}\right)^{\frac{n+1}{2}} \mathrm{K_{n+1}}\left(2\sqrt{(1+D^2)\varpi a}\right)\right)\\ \nonumber & \times\frac{|\mathcal{S}_k|!}{(|\mathcal{S}_k|-n)!n!}  \left(F_{x_i}\left(\frac{\epsilon}{P^{ow}}\right) \right)^{|\mathcal{S}_k|-n}   \left(1-F_{x_i}\left(\frac{\epsilon}{P^{ow}}\right) \right)^{n},
\end{align}
where $F_{x_i}(z)\triangleq  1- \frac{e^{-z}}{D^2z}+\frac{e^{-(1+D^2)z}}{D^2z}$, $a=\frac{1}{2}(D^2+2)$, $\varpi=\frac{n\epsilon}{\eta P^{ow}}$, and $\mathrm{K}_n(\cdot)$ denotes the modified Bessel function of the second kind.  }
\end{theorem}
\begin{proof}
Please refer to the appendix.
\end{proof}
As demonstrated by the simulations in the next section, the analytical results in Theorem \ref{theorem1} are accurate even in a moderate SNR regime. However, these results are still too involved to  be directly used for the analysis of coalition formation, which motivates the following corollary with more explicit expressions.
 \begin{corollary}\label{corollary}{\it
 When the SNR approaches  infinity, i.e. $P^{ow}\rightarrow \infty$,     an asymptotic expression for $\mathrm{P}_i$,  $i\in\mathcal{S}_k$, is
\begin{align} \label{sym snr}
\mathrm{P}_i &\rightarrow \left( \frac{\epsilon D^2}{2}+  \frac{\left( D^2+2    \right)^2  \epsilon }{4\eta }        \frac{|\mathcal{S}_k|}{(|\mathcal{S}_k|-1) }     \right)\frac{1}{P^{ow}}
\end{align}
if   $|\mathcal{S}_k|\geq 2$. And the outage probability for the users in singleton sets, i.e. $|\mathcal{S}_k|=1 $, is given by
\begin{align} \label{sym snr2}
\mathrm{P}_i &\rightarrow \left( \frac{\epsilon D^2}{2}+  \frac{ \epsilon(D^2+2) }{2\eta D^2 }  \left[    \ln \sqrt{\frac{a\epsilon}{\eta P^{ow}}}+c_0  \right.\right.\\ \nonumber &\left.\left.-  (1+D^2)^2\left(\ln \sqrt{\frac{(1+D^2)\epsilon a}{\eta P^{ow}}}+c_0\right)  \right] \right)\frac{1}{P^{ow}}.
\end{align}
 }
\end{corollary}
\begin{proof}
Please refer to the appendix.
\end{proof}\vspace{-1em}

\subsection{Characterizing   the addressed game}
Based on the above theorem and corollary, we can show that a grand coalition is always preferred.
\begin{corollary}\label{corollary 2}{\it
When the SNR approaches infinity, the optimization problem   in \eqref{original problem} becomes a canonical coalitional game with transferable utility (TU), and a grand coalition is stable.}
\end{corollary}
\begin{proof}
The asymptotic result in \eqref{sym snr} reveals   the superadditivity property of the addressed game, i.e. $v(\mathcal{S}_i)+v(\mathcal{S}_j)<v(\mathcal{S}_i\cup \mathcal{S}_j)$, for any $|\mathcal{S}_i|\geq 2$ and $|\mathcal{S}_j|\geq 2$. The result in \eqref{sym snr2} reveals that the outage probability of a user forming a singleton set decays at a rate of $\frac{\log SNR }{SNR}$. Since  a faster decaying rate of $\frac{1}{SNR}$ can be  achieved by any coalition  with a size larger than $1$, singleton sets are not preferred, and any users in singleton sets will try to join in other existing coalitions. Therefore a grand coalition is a stable point of the system.
\end{proof}

Let $\mathbf{x}^*\triangleq \{x_n^* = \frac{1}{N}v(\mathcal{N}), n \in \{1, \ldots, N\}  \}$   denote  the player payoff vector. Following Corollaries \ref{corollary} and \ref{corollary 2}, the core of the addressed game can   be characterized as follows.
\begin{corollary}\label{corollary 3}{\it
When the SNR approaches infinity, the core of the considered canonical coalitional game with TU, denoted by $\mathcal{C}$, becomes non-empty and $\mathbf{x}^*\in \mathcal{C}$. }
\end{corollary}
\begin{proof}
The payoff vector $\mathbf{x}^*$ is group-rational. And it is also individually rational since the outage probability of a user in a singleton set   decays at a slow rate of $\frac{\log SNR}{SNR}$. The proof that  $\mathbf{x}^*\in \mathcal{C}$ can be shown by   contradiction. Suppose that there is a coalition $\mathcal{S}$ that  rejects the proposed payoff allocation, which means $\sum_{i\in \mathcal{S}}x_i^*<v(\mathcal{S}) $. Recall $v(\mathcal{S}) =R\sum_{i\in \mathcal{S}}(1-\mathrm{P}_i)$. From Corollary \ref{corollary}, we first observe that the outage probability depends on only the size of the coalition, so the value of $\mathcal{S}$ can be simplified to $v(\mathcal{S}) =R|\mathcal{S}|(1-\mathrm{P}_{|\mathcal{S}|})$. Similarly for the grand set, we have $v(\mathcal{N}) =R|\mathcal{N}|(1-\mathrm{P}_{|\mathcal{N}|})$, which means $x_i^*=R(1-\mathrm{P}_{|\mathcal{N}|})$. The use of Corollary \ref{corollary} yields $\mathrm{P}_{|\mathcal{N}|}<\mathrm{P}_{|\mathcal{S}|}$, and therefore  $\sum_{i\in \mathcal{S}}x_i^*>v(\mathcal{S}) $, which contradicts the initial claim. The proof is completed.
\end{proof}

\section{Numerical Results}
In this section, simulation results are provided to demonstrate the accuracy of the developed analytical results. In Fig. \ref{fig_1}, the analytical results developed in Theorem \ref{theorem1} and the asymptotic results shown in Corollary \ref{corollary} are compared to computer simulations. Recall that the approximation steps, such as the ones in \eqref{app1} and \eqref{app3}, assume  $\frac{D^2}{P^{ow}}\rightarrow 0$. As can be seen from the figure, for a small $\mathcal{D}$, the developed analytical results  match exactly with  simulations, and for a case with a larger $\mathcal{D}$ there will be a gap which can be reduced by further increasing the SNR.  In Fig. \ref{fig_2}, the outage probability with different size   coalitions is shown as a function of SNR. As can be seen from the figure, a larger coalition is always beneficial since users experience fewer  errors and hence receive higher payoff. Particularly the case with a grand coalition achieves the best performance, which confirms Corollary \ref{corollary 2}.

\section{Conclusion}
In this paper, we have developed the outage probability for users in  a cooperative network with  an  energy harvesting relay.   In addition, the cooperation among users has been  modeled as a canonical coalitional game and the grand coalition has been  shown to be stable in the addressed scenario.

\vspace{-1.5em}
\begin{figure}[!htbp]\centering
    \epsfig{file=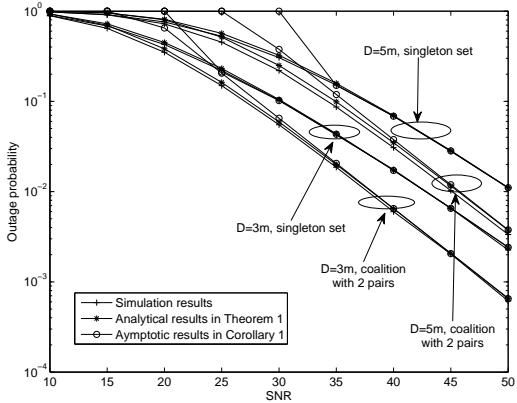, width=0.4\textwidth, clip=}\vspace{-1em}
\caption{Simulation vs analytical results. $R=0.5$ bit per channel use (BPCU).}\label{fig_1}\vspace{-2em}
\end{figure}
\vspace{-0.5em}
\begin{figure}[!htbp]\centering
    \epsfig{file=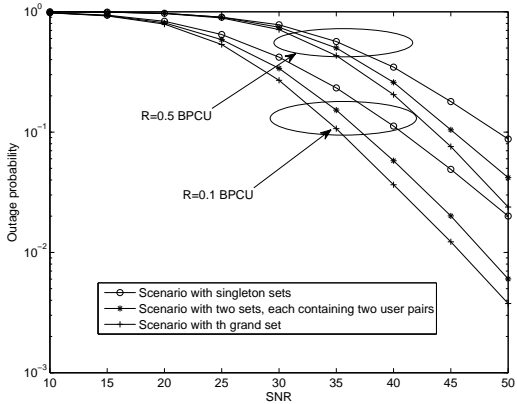, width=0.4\textwidth, clip=}\vspace{-1em}
\caption{Outage probability for coalitions with different sizes.}\label{fig_2}\vspace{-2em}
\end{figure}
\appendix
\textsl{Proof of Theorem \ref{theorem1} :}
Define $x_i\triangleq \frac{|h_i|^2}{1+d_i^2}$ and $y_i\triangleq \frac{|g_i|^2}{1+c_i^2}$.
Note that $x_i$ and $y_i$ are independent and identically   distributed (i.i.d.). $h_i$ and $g_i$ are i.i.d. complex Gaussian variables, and the density function of the distance can be found by using the fact that the locations of the nodes are uniformly distributed in  $\mathcal{D}$  \cite{Kingmanbook}.  For any area $A$ with a size of $ {\Delta}$ and  $A\in \mathcal{D}$, the distribution  of each point $W$ is $\mathrm{P}(W\in A)=\frac{\Delta}{\pi D^2}$, and the corresponding probability density function (pdf) is $p_{W}(w)=\frac{1}{\pi D^2}$. In particular, it can be shown that the cumulative distribution function (CDF) of $x_i$ is
\begin{align}\label{eq3}
F_{x_i}(z)&\triangleq \mathrm{P}\left(x_i<z\right) = \underset{  \mathcal{D}}{ \int} \left(1-e^{-(1+d_i^2)z}\right)p_{W}(w)dw,
\end{align}
where   $d_i$ is determined by  the distance between the point $W$ and the origin \cite{Wangpoor11}.
By applying   polar coordinates,  the CDF of $x_i$ can be obtained as  in Theorem \ref{theorem1}, and its pdf is
\begin{eqnarray}\label{eq2}
f_{x_i}(z) = \frac{(1+z)e^{-z} - (1+z+zD^2)e^{-(1+D^2)z} }{D^2z^2}.
\end{eqnarray}

 With   equal power allocation,  the outage probability of the $i$-th user pair, $i\in\mathcal{S}_k$, is given by{\small
\begin{align} \label{outage probability}
\mathrm{P}_i &= F_{x_i}\left( \frac{\epsilon}{P^{ow}}  \right)+\underset{Q_1}{\underbrace{\sum^{|\mathcal{S}_k|}_{n=1}\mathrm{P}\left(  i \in \tilde{\mathcal{S}_k},  y_i <\frac{\epsilon}{\tilde{P}_i^{ow}} , |\tilde{\mathcal{S}}_k|=n\right)}}.
\end{align}}
  The first probability in \eqref{outage probability}  is  $F_{x_i}\left(\frac{\epsilon}{P^{ow}}\right)$, and the second probability, $Q_1$, can be expanded as follows:
\begin{eqnarray} \nonumber
\mathrm{Q}_1 =  \sum^{|\mathcal{S}_k|}_{n=1}\frac{n}{| {\mathcal{S}}_k|}\underset{Q_{2,n}}{\underbrace{\mathrm{P}\left( y_i<\frac{\epsilon}{\tilde{P}_i^{ow}} \left| |\tilde{\mathcal{S}}_k|=n\right.\right)}}\mathrm{P} (|\tilde{\mathcal{S}}_k|=n ).
\end{eqnarray}
Since the source-relay channels are i.i.d.,  $\mathrm{P} (|\tilde{\mathcal{S}}_k|=n )$ is
\begin{eqnarray} \nonumber
\mathrm{P} (|\tilde{\mathcal{S}}_k|=n ) =\frac{|\mathcal{S}_k|! \left(F_{x_i}\left(\frac{\epsilon}{P^{ow}}\right) \right)^{|\mathcal{S}_k|-n}  }{(|\mathcal{S}_k|-n)!n!}   \left(1-F_{x_i}\left(\frac{\epsilon}{P^{ow}}\right) \right)^{n}.\label{q2 1}
\end{eqnarray}
On denoting  $z=\sum_{j\in\tilde{\mathcal{S}}_k}  x_j$,    $Q_{2,n}$ is given by
\begin{eqnarray}\label{equation integral}
\mathrm{Q}_{2,n} &=& \mathrm{P}\left(y_i<\frac{\epsilon}{\tilde{P}_i^{ow}} \left| |\tilde{\mathcal{S}}_k|=n\right.\right)  \\ \nonumber  &=& \int^{\infty}_{n \frac{\epsilon}{P^{ow}}}F_{x_i}\left(\frac{\epsilon}{\frac{\eta P^{ow} }{|\tilde{\mathcal{S}}_k|} z - \epsilon \eta} \right)f_{z||\tilde{\mathcal{S}}_k|=n}(z)dz.
\end{eqnarray}
Directly finding the pdf of $z$ conditioned on  $|\tilde{\mathcal{S}}_k|=n$,  $f_{z||\tilde{\mathcal{S}}_k|=n}(z)$,  is difficult. To obtain an approximation of $\mathrm{Q}_{2,n}$, first define $\gamma(z,P^{ow})\triangleq\epsilon\left(\frac{\eta P^{ow} }{|\tilde{\mathcal{S}}_k|} z - \epsilon \eta\right)^{-1}$, and observe that  $F_{x_i}\left(\gamma(z,P^{ow})\right)$ exhibits the following property:
\begin{align}\nonumber
&F_{x_i}\left(\gamma(z,P^{ow}) \right)  \approx  1- \frac{\left(1-\gamma(z,P^{ow})+\frac{1}{2}\gamma^2(z,P^{ow})\right) }{D^2\gamma(z,P^{ow})} \\ \nonumber & +\frac{\left(1-\gamma(z,P^{ow})(1+D^2) +\frac{1}{2}\gamma^2(z,P^{ow})(1+D^2)^2 \right)}{D^2\gamma(z,P^{ow})}
  \\ \label{app1}&=\frac{1}{2}\gamma(z,P^{ow})(D^2+2)\rightarrow 0,
\end{align}
for any fixed $z$ and a sufficiently large $P^{ow}$. Or in other words, when the transmission  power is large,  the probability   quickly decreases  to zero, i.e.  $F_{x_i}\left(\gamma(z,P^{ow})\right)\rightarrow 0$, for $z>e$ and $e$ is small.

This means that for the integral in \eqref{equation integral}, a good  approximation of $f_{z||\tilde{\mathcal{S}}_k|=n}(z)$, denoted by $\tilde{f}_{z||\tilde{\mathcal{S}}_k|=n}(z)$, needs to ensure  $\tilde{f}_{z||\tilde{\mathcal{S}}_k|=n}(z)=f_{z||\tilde{\mathcal{S}}_k|=n}(z)$ for $z\rightarrow 0$, and a slight difference between two functions for a large $z$ has an insignificant impact on the integral since  $F_{x_i}\left(\gamma(z,P^{ow})\right)\rightarrow 0$ for a large $z$. Recall that $z=\sum_{j\in\tilde{\mathcal{S}}_k}  x_j$, and the above rationale   motivates us to find an approximation of $f_{x_j}(z)$ for $z \rightarrow 0$. Particularly when $z\rightarrow 0$, the pdf of $x_j$ can be approximated  as
\begin{eqnarray} \nonumber
f_{x_j}(z) &\approx& \frac{1}{D^2z^2}(1+z)\left(1-z+\frac{z^2}{2}\right)  -  \frac{1}{D^2z^2}   (1+z\\ \nonumber &&+zD^2)\left(1- (1+D^2)z+\frac{1}{2} (1+D^2)^2z^2 \right)\\ \nonumber &=& \left(\frac{1}{2}D^2+1\right) +o(z^2).
\end{eqnarray}
Directly using the approximation $ {f}_{x_j}(z)\approx \frac{1}{2}D^2+1$ means that the range of the variable is limited as $\begin{bmatrix} 0 , \frac{2}{D^2+2} \end{bmatrix}$ to ensure the CDF normalized.  Instead we use the following approximation:
\begin{eqnarray}
f_{x_j}(z) \approx \left(\frac{1}{2}D^2 +1\right)e^{-\left(\frac{1}{2}D^2+1\right)z},
\end{eqnarray}
an exponential distribution with the parameter $ \left(\frac{1}{2}D^2+1\right)$. As can be seen in Section IV, this approximation is tightly matched  with  the simulation results even in the moderate SNR regime.

Conditioned on $x_j>\frac{\epsilon}{P^{ow}}$, the Laplace transform of $x_j$ is
\begin{align}\nonumber
\mathcal{L}(f_{x_j|x_j>\frac{\epsilon}{P^{ow}}}(t))&= \frac{1}{1-F_{x_j}(\frac{\epsilon}{P^{ow}})}\int^{\infty}_{\frac{\epsilon}{P^{ow}}}e^{-st}f_{x_j}(t)dt
\\\nonumber &\approx   \left(\frac{1}{2}D^2+1\right) e^{\frac{D^2\epsilon}{2P^{ow}}} \frac{e^{-(s+\left(\frac{1}{2}D^2+1\right))\frac{\epsilon}{P^{ow}}}}{s+\left(\frac{1}{2}D^2+1\right)}.\end{align}
The Laplace transform for the pdf of the sum  is
\begin{eqnarray}\nonumber
\mathcal{L}(f_{z||\tilde{\mathcal{S}}_k|=n}(z))\approx  \left(\frac{1}{2}D^2+1\right)^n \frac{e^{-\frac{n\epsilon s}{P^{ow}}}}{\left(s+\left(\frac{1}{2}D^2+1\right)\right)^n},
\end{eqnarray}
which yields the pdf of $z$ conditioned on $|\tilde{\mathcal{S}}_k|=n$ as follows:
\begin{eqnarray}\nonumber
f_{z||\tilde{\mathcal{S}}_k|=n}(z)\approx   \frac{a^n  \left(z-n\frac{\epsilon}{P^{ow}}\right)^{n-1}e^{-a\left(z-n\frac{\epsilon}{P^{ow}}\right)}  }{(n-1)!},
\end{eqnarray}
for $\frac{n\epsilon}{P^{ow}}\leq z \leq \infty.$
On substituting it into the probability of $Q_{2,n}$ we obtain the following:
\begin{eqnarray}\nonumber
\mathrm{Q}_{2,n} &=&  \int^{\infty}_{n\frac{\epsilon}{P^{ow}}}F_{x_i}\left(\frac{n\epsilon}{\eta P^{ow}}\frac{1}{ z -n\frac{\epsilon}{  P^{ow} }    } \right)f_{z||\tilde{\mathcal{S}}_k|=n}(z)dz\\ \nonumber
&\approx&    \frac{a^n}{(n-1)!} \int^{\infty}_{0}F_{x_i}\left(\frac{n\epsilon}{\eta P^{ow}}\frac{1}{t   } \right)e^{-at}t^{n-1}dt .
\end{eqnarray}
After some algebra  manipulations, the probability $Q_{2,n}$ can be expressed as
\begin{align}\label{eq q2}
\mathrm{Q}_{2,n} &\approx  \frac{ a^{n}  }{(n-1)! }  \left((n-1)!a^{-n} \right. \\ \nonumber &-\frac{2}{D^2\varpi}  \left(\frac{\varpi}{a}\right)^{\frac{n+1}{2}} \mathrm{K_{n+1}}\left(2\sqrt{\varpi a}\right)+\frac{2}{D^2\varpi}\\   &\times\left.  \left(\frac{(1+D^2)\varpi}{a}\right)^{\frac{n+1}{2}} \mathrm{K_{n+1}}\left(2\sqrt{(1+D^2)\varpi a}\right)\right). \nonumber
 \end{align}
By combining \eqref{outage probability}, \eqref{q2 1}  and \eqref{eq q2}, the theorem is proved.   \hspace{\fill}$\blacksquare$\newline

\textsl{Proof of Corollary \ref{corollary} :}
By applying the series representation of Bessel functions,  $x^n\mathbf{K}_n(x)$ can be approximated as \cite{GRADSHTEYN}
\begin{eqnarray}\nonumber
x^2\mathbf{K}_2(x) = \frac{1}{2}\left(4-x^2 \right) +\frac{x^4}{8}\left(-\ln \frac{x}{2} -c_0\right)+o(x^6\ln x) \label{approximation},
\end{eqnarray}
for $x\rightarrow 0$, where $c_0=C-\frac{3}{4}$ and $C$ is Euler's constant. And
\begin{eqnarray}\nonumber
x^n\mathbf{K}_n(x) =  \frac{1}{2}\sum^{n-1}_{l=0}\frac{(-1)^l(n-l-1)!}{l!} \frac{x^{2l}}{2^{2l-n}} +o(x^{2n}\ln x) \label{approximation1},
\end{eqnarray}
for $n\geq 3$.
For  $n\geq 3$,  $\mathrm{Q}_{2,n}$ can be approximated  as follows:
 \begin{eqnarray}\nonumber
\mathrm{Q}_{2,n} \approx
 \frac{a^{n}  }{(n-1)!} \left((n-1)!a^{-n}-\frac{1}{D^2\varpi 2^{n+1} a^{n+1}}  \right.  \\ \nonumber \left.\times \sum^{n}_{j=0}\frac{(-1)^j(n-j)!}{j!} \frac{ \left(2\sqrt{\varpi a}\right)^{2j}}{2^{2j-n-1}} +\frac{1}{D^2\varpi 2^{n+1} a^{n+1}}\right. \\ \nonumber \left. \times\sum^{n}_{j=0}\frac{(-1)^j(n-j)!}{j!} \frac{ \left(2\sqrt{(1+D^2)\varpi a}\right)^{2j}}{2^{2j-n-1}}  \right)
\approx
 \frac{a\left( D^2+2    \right) }{2(n-1)}  \varpi.
\end{eqnarray}
And for $n=2$, we can obtain  the following:
\begin{align}\nonumber
\mathrm{Q}_{2,2}  &\approx   \frac{ \varpi(D^2+2) }{2D^2 }  \left(   \left(\ln \sqrt{\varpi a}+c_0\right)    -  (1+D^2)^2\right. \\ \nonumber &\left.(\ln \sqrt{(1+D^2)\varpi a}+c_0)  \right).
\end{align}
Note that  $\mathrm{Q}_{2,n}$ for  $n\geq 3$ decays  at a rate of $\frac{1}{P^{ow}}$ and  $\mathrm{Q}_{2,2}$ has a much slower decay rate of $\frac{\log P^{^{ow}}}{P^{ow}}$.

By applying the above approximations, when $|\mathcal{S}_k|\geq 2$, the outage probability experienced by a user in $\mathcal{S}_k$ is asymptotically equivalent to the following expression:
\begin{eqnarray}\nonumber
\mathrm{P}_i &=& F_{x_i}\left(\frac{ \epsilon}{P^{ow}}\right)+\sum^{|\mathcal{S}_k|}_{n=1}\frac{n}{| {\mathcal{S}}_k|}Q_{2,n}\cdot \mathrm{P} (|\tilde{\mathcal{S}}_k|=n )\\ \nonumber &\rightarrow&\frac{\epsilon D^2}{2P^{ow}}+\sum^{|\mathcal{S}_k|}_{n=1}\frac{n}{| {\mathcal{S}}_k|}  \frac{a\left( D^2+2    \right)  }{2(n-1)}   \left(    \frac{n\epsilon}{\eta P^{ow}}    \right)\\ \nonumber &&\times \frac{|\mathcal{S}_k|!}{(|\mathcal{S}_k|-n)!n!}  \left( \frac{\epsilon}{2P^{ow}} \right)^{|\mathcal{S}_k|-n}
\\   &=&\frac{\epsilon D^2}{2P^{ow}}+  \frac{a \left( D^2+2    \right)  }{2 (|\mathcal{S}_k|-1) }        \frac{|\mathcal{S}_k|\epsilon}{\eta }     \frac{1}{P^{ow}},\label{app3}
\end{eqnarray}
and the first part of the corollary is proved. The second part of the corollary  can be proved using similar steps.  \hspace{\fill}$\blacksquare$\newline
\vspace{-1em}

 \bibliographystyle{IEEEtran}
\bibliography{IEEEfull,trasfer}
\end{document}